\documentclass[times,doublespace]{rncauth}
\usepackage{moreverb}
\usepackage{rawfonts,graphicx,amsfonts,amssymb,psfrag,flushend}

\graphicspath{{simulation_graph/}}

\newtheorem{thm}{Theorem}

\newtheorem{rmk}{Remark}
\newtheorem{df}{Definition}

\begin{document}

\title{Guaranteed-cost consensus for multiagent networks with Lipschitz nonlinear dynamics and switching topologies}

\author{Jianxiang~Xi\affil{1}, Zhiliang~Fan\affil{1}, Hao~Liu\affil{2,}\corrauth, Tang~Zheng\affil{1}}
\address{\hspace{12em}\affilnum{1}High-Tech Institute of Xi'an, Xi'an, 710025,
P.R. China\\ \affilnum{2}School of Astronautics, Beihang University, Beijing, 100191, P.R. China \vspace{-2em}} \corraddr{liuhao13@buaa.edu.cn;~xijx07@mails.tsinghua.edu.cn}

\cgsn{This work was supported by the National Natural Science Foundation of China under Grants 61374054, 61503012, 61503009, 61333011 and 61421063 and by Shaanxi Province Natural Science Foundation Research Projection under Grants 2016JM6014, also supported by Innovation Foundation of High-Tech Institute of Xi'an under Grants 2015ZZDJJ03. The authors would like to thank Minghua Liu for providing numerical analysis and simulation.}{}

\begin{abstract}
Guaranteed-cost consensus for high-order nonlinear multi-agent networks with switching topologies is investigated. By constructing a time-varying nonsingular matrix with a specific structure, the whole dynamics of multi-agent networks is decomposed into the consensus and disagreement parts with nonlinear terms, which is the key challenge to be dealt with. An explicit expression of the consensus dynamics, which contains the nonlinear term, is given and its initial state is determined. Furthermore, by the structure property of the time-varying nonsingular transformation matrix and the Lipschitz condition, the impacts of the nonlinear term on the disagreement dynamics are linearized and the gain matrix of the consensus protocol is determined on the basis of the Riccati equation. Moreover, an approach to minimize the guaranteed cost is given in terms of linear matrix inequalities. Finally, the numerical simulation is shown to demonstrate the effectiveness of theoretical results.
\end{abstract}


\maketitle \vspace{-2em}

\section{Introduction }\label{section1} \vspace{-0em}
Consensus is a typical collection behavior of multi-agent networks consisting of a number of autonomous dynamic agents and has been extensively investigated recently due to its wide applications in different fields such as formation and containment control for unmanned aerial vehicles \cite{c1}-\cite{c4}, synchronization control for sensor networks \cite{c5}-\cite{c6} and reconfiguration for spacecraft clusters \cite{c7}-\cite{c8}, {\it et al}. For leadless multi-agent networks, the whole dynamics contains two parts: the consensus dynamics and the disagreement dynamics, which describe the macroscopical and microcosmic behaviors of multi-agent networks, respectively. The consensus dynamics is often described by the consensus function and the disagreement dynamics is used to determine the consensus and consensualization criteria.\par

According to the dynamics of each agent, multi-agent networks can be classified three types: first-order ones, second-order ones and high-order ones. The dynamics of each agent in first-order and second-order multi-agent networks is usually described as the first-order and second-order integrators, respectively, whose consensus analysis and design problems can be simplified by these structure features (see \cite{c9}-\cite{c14} and references therein). Each agent in high-order multi-agent networks is often modeled by a general high-order linear system, whose consensus and consensualization criteria are more difficult to be determined since each agent does not have specific structure features compared with first-order and second-order integrators. Some important and interesting works about high-order linear multi-agent networks were finished in \cite{c15}-\cite{c184}, where the optimization performance was not considered.\par

In practical applications of multi-agent networks, each agent may have limited energy supply to perform certain tasks, such as sensing, communication and movement, {\it et al}. Meanwhile, consensus regulation is required to satisfy some performance indexes. For examples, when a multiple mobile autonomous vehicle network performs a specific patrol task, the distance performance is very important due to utility maximization. Thus, it is a crucial challenge to realize the tradeoff design between consensus regulation performance and energy consumption, which can usually be modeled as optimal or suboptimal consensus problems. For first-order multi-agent networks, the optimal consensus criterion was proposed in \cite{c19}, where it is essentially required that the interaction topology is a completed graph, which means that all nonzero eigenvalues of the associated Laplacian matrix are identical. For second-order multi-agent networks, Guan {\it et al}. \cite{c20} considered consensus regulation performance, but energy consumption was not dealt with, and Wang {\it et al}. \cite{c21} constructed a linear quadratic optimization index and proposed suboptimal consensus criteria.\par

For high-order multi-agent networks, optimal consensus is difficult to be achieved due to complex structures and guaranteed-cost consensus is proposed to realize suboptimal control. Linear matrix inequality (LMI) criteria for guaranteed-cost consensualization were proposed in \cite{c22} and \cite{c23}, where the dimensions of all the variables are associated with the number of agents, so the computation complexity greatly increases as the number of agents increases. This shortcoming was overcome in \cite{c24}, where the dimensions of all the variables in LMI guaranteed-cost consensus criteria are equal to the one of the dynamics of each agent. To the best of our knowledge, guaranteed-cost consensus of multi-agent networks with nonlinear dynamics and switching topologies is not comprehensively investigated and the following three challenging problems are still open: (i) How to determine the consensus function when the nonlinear dynamics is involved; (ii) How to linearize the impacts of nonlinearity on disagreement dynamics and to determine the gain matrix of the consensus protocol with switching topologies; (iii) How to obtain the minimum guaranteed cost.\par

The current paper focuses on guaranteed-cost consensus control for high-order multi-agent networks with Lipschitz nonlinear dynamics and switching topologies. By constructing a time-varying orthonormal matrix with a specific structure, the whole dynamics of multi-agent networks is decomposed into two subsystems with nonlinear dynamics, which are used to determine the consensus function and the consensualization criterion, respectively. Then, it is shown that the consensus dynamics is intrinsically coincided with the own dynamics of each agent, but the initial states are different. Furthermore, it is revealed that the product of the time-varying part of the transformation matrix and its transpose is equivalent to the Laplacian matrix of a complete graph. Based on this property and the Lipschitz condition, a sufficient condition for guaranteed-cost consensualization is presented in terms of the Riccati equation, the dimension of whose variable is independent of the number of agents, and an upper bound of the guaranteed cost is determined. Moreover, an approach is proposed to minimize the guaranteed cost.\par

Compared with the existing works about guaranteed-cost consensus of high-order multi-agent networks, the current paper has the following four novel features. Firstly, the current paper studies the impacts of nonlinear dynamics on guaranteed-cost consensus, but the literatures \cite{c22}-\cite{c24} did not. Secondly, the current paper determines an explicit expression of the consensus dynamics and its initial states. No approach was given to determine the consensus function in \cite{c22} and \cite{c23}, and the method in \cite{c24} cannot be used to investigate nonlinear cases. Thirdly, the current paper proposes an approach to linearize the nonlinear factor on disagreement dynamics, while the methods in \cite{c22}-\cite{c24} are no longer valid to deal with nonlinear dynamics. Fourthly, the current paper determines the minimum guaranteed cost under mild assumptions, but the guaranteed costs given in \cite{c22}-\cite{c24} may not be minimum.\par

The current paper is organized as follows. In Section \ref{section2}, based on graph theory, the problem description is given. Section \ref{section3} presents an approach to give the explicit expression of the consensus dynamics and to determine its initial states. In Section \ref{section4}, sufficient conditions for guaranteed-cost consensualization are proposed, and the approaches to determine the guaranteed cost and the minimum guaranteed cost are presented, respectively. A numerical example is shown to demonstrate theoretical results in Section \ref{section5}. Finally, concluding remarks are given in Section \ref{section6} and basic concepts and conclusions on graph theory are given in Section \ref{section7} .\par

{\it Notations}: ${\mathbb{R}^d}$ is the $d$-dimensional real column vector space and ${\mathbb{R}^{d \times p}}$ is the set of $d \times p$ dimensional real matrices, where $d$ and $p$ are positive integers. ${I_N}$ and ${{\bf{1}}_N}$ represent the $N$-dimensional real identity matrix and the column vector with all components 1, respectively. $0$ is used to represent the zero number and the zero vector with a compatible dimension. The notation $ \otimes $ denotes the Kronecker product. The symbol * in the matrix stands for the symmetric term. ${P^T} = P > {\rm{0}}$ and ${P^T} = P < {\rm{0}}$ mean that the symmetric matrix $P$ is positive definite and negative definite, respectively. ${\textup{trace}}\left\{ P \right\}$ denotes the trace of the matrix $P$. \vspace{-2em}\par

\section{Problem description}\label{section2} \vspace{-0.5em}
Consider multi-agent networks with $N$ identical nonlinear agents, where the $i$th agent is modeled by
\begin{eqnarray}\label{1}
{\dot x_i}(t) = A{x_i}(t) + f({x_i}(t)) + B{u_i}(t){\rm{ }}\left( {i \in \left\{ {1,2, \cdots ,N} \right\}} \right),
\end{eqnarray}
where $A \in {\mathbb{R}^{d \times d}},$ $B \in {\mathbb{R}^{d \times p}},$ ${x_i}(t)$ and ${u_i}(t)$ are the state and the control input, respectively and the nonlinear function $f:{\mathbb{R}^d} \times \left[ {0, + \infty } \right) \to {\mathbb{R}^d}$ is continuous and differentiable, which satisfies the following Lipschitz condition with a Lipschitz constant $\gamma  > 0$:
\[
\left\| {f(y) - f(z)} \right\| \le \gamma \left\| {y - z} \right\|,
\]
where $y$ and $z$ are any vectors with compatible dimensions. Multi-agent networks with the above Lipschitz nonlinear dynamics extensively exist in practical applications. For instants, the sinusoidal term in multiple manipulator systems is globally Lipschitz as shown in \cite{c25} and \cite{c26}. Actually, the Lipschitz constant $\gamma  > 0$ is a low bound of the distance to unobservability of $(A^T,B^T)$ and $(A,B)$ should be controllable to guarantee that the distance to unobservability of $(A^T,B^T)$ is positive. As shown in \cite{c27}, it is sufficient for linear multi-agent networks to achieve consensus that $(A,B)$ is controllable, but this conclusion does not hold when the Lipschitz nonlinear dynamics exists.  \par

Let $\eta \subset \mathbb{N}$ be an index of the set $\kappa $ of some undirected graphs, where $\mathbb{N}$ is the natural number set. The time function $\sigma (t):\left[ {\left. {0, + \infty } \right)} \right. \to \eta $ denotes the switching signal. The interaction topologies among the agents of multi-agent network (\ref{1}) are randomly switching among undirected graphs belonging to the set $\kappa $ and can be described by ${G_{\sigma (t)}}{\rm{ = }}\left( {V\left( {{G_{\sigma (0)}}} \right),E\left( {{G_{\sigma (t)}}} \right)} \right)$, where $V\left( {{G_{\sigma (0)}}} \right)$ is a nonempty finite time-invariant set of nodes and $E\left( {{G_{\sigma (t)}}} \right) \subseteq V\left( {{G_{\sigma (t)}}} \right) \times V\left( {{G_{\sigma (t)}}} \right)$ is a set of edges. The node ${v_i} \in {G_{\sigma (0)}}$ presents agent $i$ and the edge $\left( {{v_j},{v_i}} \right) \in E\left( {{G_{\sigma (t)}}} \right)$ denotes that there exists an interaction channel from agent $j$ to agent $i$. ${G_{\sigma (t)}}$ is said to be connected if there at least exists one node having an undirected path to every other nodes at any time $t$. Moreover, it is assumed that switching times $\{ {t_i}:i = 0,1,2, \cdots \} $ of the switching signal $\sigma (t)$ satisfy that ${t_k} - {t_{k{\rm{ - }}1}} \ge {T_{\rm{d}}}{\rm{ }}~(\forall k \ge 1)$ for a positive constant ${T_{\rm{d}}}$.\par

The following consensus protocol is applied for agent $i$ to collect the state information of its neighboring agents
\begin{eqnarray}\label{2}
{u_i}(t) = K\sum\limits_{j \in {N_{\sigma (t),i}}} {{w_{\sigma (t),ij}}\left( {{x_j}(t) - {x_i}(t)} \right)} ,
\end{eqnarray}
where $K \in {\mathbb{R}^{p \times d}}$ is the gain matrix, ${w_{\sigma (t),ij}}$ is the weight of the interaction channel $({v_j},{v_i})$ with ${w_{\sigma (t),ii}} = 0$, ${w_{\sigma (t),ji}} = {w_{\sigma (t),ij}} \ge 0$ and ${w_{\sigma (t),ij}} > 0$ if $({v_j},{v_i}) \in E\left( {{G_{\sigma (t)}}} \right)$, and ${N_{\sigma (t),i}}(t) = \left\{ {j:({v_j},{v_i}) \in E\left( {{G_{\sigma (t)}}} \right)} \right\}$ represents the index of the neighbor set of vertex ${v_i}$. Furthermore, for given symmetric and positive matrices $Q$ and $R$, consider the following linear quadratic cost function
\begin{eqnarray}\label{3}
{J_C} = \int_0^{ + \infty } {\left( {{J_{Cu}}(t) + {J_{Cx}}(t)} \right){\rm{d}}t} ,
\end{eqnarray}
where
\[{J_{Cu}}(t) = \sum\limits_{i = 1}^N {u_i^T(t)R{u_i}(t)} ,\]\vspace{-0.25em}
\[{J_{Cx}} = \sum\limits_{i = 1}^N {\sum\limits_{j \in {N_{\sigma (t),i}}} {{w_{\sigma (t),ij}}{{\left( {{x_j}(t) - {x_i}(t)} \right)}^T}Q\left( {{x_j}(t) - {x_i}(t)} \right)} } ,\]
which are called the energy consumption term and the consensus regulation term, respectively.\par

Let ${W_{\sigma (t)}} = \left[ {{w_{\sigma (t),ij}}} \right] \in {\mathbb{R}^{N \times N}}$ and ${D_{\sigma (t)}} = {\rm{diag}}\left\{ {\sum\nolimits_{j = 1}^N {{w_{\sigma (t),ij}}} ,i = 1,2, \cdots ,N} \right\}$ be the adjacency matrix and in-degree matrix of the interaction topology, respectively, then the matrix ${L_{\sigma (t)}} = {D_{\sigma (t)}} - {W_{\sigma (t)}}$ denotes the Laplacian matrix of ${G_{\sigma (t)}}$, which satisfies that ${L_{\sigma (t)}}{{\bf{1}}_N} = 0$. It should be pointed out that ${L_{\sigma (t)}}$ is piecewise continuous since interaction topologies of multi-agent network (\ref{1}) are switching with ${T_{\rm{d}}} > 0$. Let
\[x(t) = {\left[ {x_1^T(t),x_2^T(t), \cdots ,x_N^T(t)} \right]^T}{\rm{ ,}}\]\vspace{-1.25em}
\[F(x(t)) = {\left[ {f{{\left( {{x_1}(t)} \right)}^T},f{{\left( {{x_2}(t)} \right)}^T}, \cdots ,f{{\left( {{x_N}(t)} \right)}^T}} \right]^T}{\rm{ ,}}\]
then the dynamics of multi-agent network (\ref{1}) with protocol (\ref{2}) can be written as
\begin{eqnarray}\label{4}
\dot x(t) = \left( {{I_N} \otimes A - {L_{\sigma (t)}} \otimes BK} \right)x(t) + F(x(t)).
\end{eqnarray}\par

In the following, the definitions of the guaranteed-cost consensus and consensualization of multi-agent networks are given, respectively, which are used to realize the suboptimal tradeoff design between consensus regulation performance and energy consumption.\par

\begin{df} \label{definition1}
Multi-agent network (\ref{4}) is said to achieve guaranteed-cost consensus if there exist a positive constant $\beta $ and a vector-valued function $c(t)$ such that ${\lim _{t \to  + \infty }}\left( {x(t) - {{\bf{1}}_N} \otimes c(t)} \right) = 0$ and ${J_C} \le \beta$ for any bounded initial state $x(0)$, where $\beta $ and $c(t)$ are said to be the guaranteed cost and the consensus function, respectively.
\end{df}

\begin{df} \label{definition1}
Multi-agent network (\ref{1}) is said to be guaranteed-cost consensualizable by protocol (\ref{2}) if there exists a gain matrix $K$ such that it achieves guaranteed-cost consensus.
\end{df}

The current paper mainly focuses on the following three guaranteed-cost consensus problems for nonlinear multi-agent networks with switching connected topologies: (i) How to determine the impacts of nonlinear dynamics on the consensus function; (ii)  How to design the gain matrix $K$ such that multi-agent network (\ref{4}) achieves guaranteed-cost consensus; (iii) How to determine and minimize the guaranteed cost.\vspace{-1.5em}

\section{Consensus functions}\label{section3}\vspace{-0.5em}
In this section, an explicit expression of the consensus dynamics and its initial states are given, which can be used to determine the consensus function, and the impacts of the nonlinear dynamics and switching topologies on the consensus function are shown.\par

Because it is assumed that the communication topology ${G_{\sigma (t)}}$ is undirected and connected, the Laplacian matrix ${L_{\sigma (t)}}$ is symmetric and zero is its simple eigenvalue. Hence, there exists an orthonormal matrix ${U_{\sigma (t)}} = \left[ {{u_1},{u_{\sigma (t),2}}, \cdots ,{u_{\sigma (t),N}}} \right]$ with ${{{u_1} = {{\bf{1}}_N}} \mathord{\left/{\vphantom {{{u_1} = {{\bf{1}}_N}} {\sqrt N }}} \right.\kern-\nulldelimiterspace} {\sqrt N }}$ such that
\begin{eqnarray}\label{5}
U_{\sigma (t)}^T{L_{\sigma (t)}}{U_{\sigma (t)}} = {\Lambda _{\sigma (t)}} = {\rm{diag}}\left\{ {{\lambda _{\sigma (t),1}},{\lambda _{\sigma (t),2}}, \cdots ,{\lambda _{\sigma (t),N}}} \right\},
\end{eqnarray}
where $0 = {\lambda _{\sigma (t),1}} < {\lambda _{\sigma (t),2}} \le  \cdots  \le {\lambda _{\sigma (t),N}}$ are the eigenvalues of ${L_{\sigma (t)}}$. Let $\tilde x(t) = \left( {U_{\sigma (t)}^T \otimes {I_d}} \right)x(t) = {\left[ {\tilde x_1^T(t),\tilde x_2^T(t), \cdots ,\tilde x_N^T(t)} \right]^T}{\rm{ }}$. Since ${U_{\sigma (t)}}$ is piecewise continuous and is constant at the switching interval, multi-agent network (\ref{4}) can be transformed into
\begin{eqnarray}\label{6}
{\dot {\tilde x}}(t) = \left( {{I_N} \otimes A - {\Lambda _{\sigma (t)}} \otimes BK} \right)\tilde x(t) + \left( {U_{\sigma (t)}^T \otimes {I_d}} \right)F(x(t)).
\end{eqnarray}
Let ${e_i}$ $(i \in {\rm{\{ }}1,2, \cdots ,N{\rm{\} }})$ be an $N$-dimensional column vector with the $i$th element 1 and 0 elsewhere and ${\bar U_{\sigma (t)}} = \left[ {{u_{\sigma (t),2}}, u_{\sigma (t),3},\cdots ,{u_{\sigma (t),N}}} \right]$, then one has
\begin{eqnarray}\label{7}
\left( {e_1^T \otimes {I_d}} \right)\left( {U_{\sigma (t)}^T \otimes {I_d}} \right) = u_1^T \otimes {I_d},
\end{eqnarray}
\begin{eqnarray}\label{8}
\left[ {0,{I_{(N - 1)d}}} \right]\left( {U_{\sigma (t)}^T \otimes {I_d}} \right) = \bar U_{\sigma (t)}^T \otimes {I_d}.
\end{eqnarray}
Let $\varsigma (t) = {\left[ {\tilde x_2^T(t),\tilde x_3^T(t), \cdots ,\tilde x_N^T(t)} \right]^T}{\rm{ }}$ and ${\tilde \Lambda _{\sigma (t)}} = {\rm{diag}}\left\{ {{\lambda _{\sigma (t),2}},{\lambda _{\sigma (t),3}}, \cdots ,{\lambda _{\sigma (t),N}}} \right\}$, then it can be derived from (\ref{6}) to (\ref{8}) that
\begin{eqnarray}\label{9}
{\dot {\tilde x}_{1}}(t)= A{{\tilde x}_1}(t) + \frac{1}{{\sqrt N }}\sum\limits_{i = 1}^N {f({x_i}(t))} ,
\end{eqnarray}\vspace{-1.5em}
\begin{eqnarray}\label{10}
\dot \varsigma (t) = \left( {{I_{N - 1}} \otimes A - {{\tilde \Lambda }_{\sigma (t)}} \otimes BK} \right)\varsigma (t) + \left( {\bar U_{\sigma (t)}^T \otimes {I_d}} \right)F(x(t)).
\end{eqnarray}\par

Actually, subsystems (\ref{9}) and (\ref{10}) describe the consensus and disagreement dynamics of multi-agent network (\ref{4}). It can be found that both (\ref{9}) and (\ref{10}) contain the nonlinear terms $f({x_i}(t)){\rm{ }}(i = 1,2, \cdots ,N)$ , which means that the nonsingular transformation does not completely decompose the whole dynamics of multi-agent network (\ref{4}) into the consensus and disagreement dynamics due to the influence of the nonlinear dynamics of each agent. The following theorem presents an approach to determine the consensus dynamics and its initial state.\vspace{0.5em}

\begin{thm} \label{theorem1}
If multi-agent network (\ref{4}) achieves consensus, then the consensus function $c(t)$ satisfies that
\[\dot c(t) = Ac(t) + f\left( {c(t)} \right),\]
where
\[c(0) = \frac{1}{N}\sum\limits_{i = 1}^N {{x_i}(0)} .\]
\end{thm}
\begin{proof}
Due to
\[{e_1} \otimes {\tilde x_1}(t) = {\left[ {\tilde x_1^T(t),0, \cdots ,0} \right]^T},\]
one can set that
\begin{eqnarray}\label{11}
{x_c}(t) \buildrel \Delta \over = \left( {{U_{\sigma (t)}} \otimes {I_d}} \right){\left[ {\tilde x_1^T(t),0, \cdots ,0} \right]^T} = \frac{1}{{\sqrt N }}{{\bf{1}}_N} \otimes {\tilde x_1}(t).
\end{eqnarray}
Since
\[\sum\limits_{i = 2}^N {{e_i} \otimes {{\tilde x}_i}(t)}  = {\left[ {0,\tilde x_2^T(t), \cdots ,\tilde x_N^T(t)} \right]^T},\]
it can be obtained that
\begin{eqnarray}\label{12}
{x_{\bar c}}(t) \buildrel \Delta \over = \left( {{U_{\sigma (t)}} \otimes {I_d}} \right){\left[ {0,\tilde x_2^T(t), \cdots ,\tilde x_N^T(t)} \right]^T} = \sum\limits_{i = 2}^N {{u_{\sigma (t),i}} \otimes {{\tilde x}_i}(t)}.
\end{eqnarray}
Because ${u_{\sigma (t),i}}$ $(i = 1,2, \cdots ,N)$ are linearly independent, ${x_c}(t){\rm{ }}$ and ${x_{\bar c}}(t){\rm{ }}$ are linearly independent by (\ref{11}) and (\ref{12}). Due to
\[\left( {U_{\sigma (t)}^T \otimes {I_d}} \right)x(t) = {\left[ {\tilde x_1^T(t),\tilde x_2^T(t), \cdots ,\tilde x_N^T(t)} \right]^T}{\rm{ }},\]
one can show that
\[x(t) = {x_c}(t) + {x_{\bar c}}(t){\rm{ }}.\]
According to the structure of ${x_c}(t)$ given in (\ref{11}), multi-agent network (\ref{4}) achieves consensus if and only if ${\lim _{t \to  + \infty }}{\left[ {\tilde x_2^T(t),\tilde x_3^T(t), \cdots ,\tilde x_N^T(t)} \right]^T} = 0$; that is, ${\lim _{t \to  + \infty }}\varsigma (t) = 0$. In this case, ${{{{\tilde x}_1}(t)} \mathord{\left/{\vphantom {{{{\tilde x}_1}(t)} {\sqrt N }}} \right.\kern-\nulldelimiterspace} {\sqrt N }}$ is a valid candidate of the consensus function.\par
If multi-agent network (\ref{4}) achieves consensus, then one can show that
\[\mathop {\lim }\limits_{t \to  + \infty } \left( {x(t) - {x_c}(t)} \right) = \mathop {\lim }\limits_{t \to  + \infty } \left( {x(t) - \frac{1}{{\sqrt N }}{{\bf{1}}_N} \otimes {{\tilde x}_1}(t)} \right) = 0.\]
Hence, one can obtain that
\[\mathop {\lim }\limits_{t \to  + \infty } \left( {c(t) - \frac{1}{{\sqrt N }}{{\tilde x}_1}(t)} \right) = 0.\]
Since multi-agent network (\ref{4}) achieves consensus, one has ${\lim _{t \to  + \infty }}\left( {{x_i}(t) - c(t)} \right) = 0$ $(i = 1,2, \cdots ,N)$, which means that
\[\mathop {\lim }\limits_{t \to  + \infty } \left( {f({x_i}(t)) - f\left( {\frac{1}{{\sqrt N }}{{\tilde x}_1}(t)} \right)} \right) = 0{\rm{ }} \hspace{2pt} (i = 1,2, \cdots ,N).\]
Let
\[\upsilon (t) = \frac{1}{{\sqrt N }}{\tilde x_1}(t).\]
Then, one can derive by (\ref{9}) that
\[\dot \upsilon (t) = A\upsilon (t) + f\left( {\upsilon (t)} \right).\]
Due to $\tilde x(t) = \left( {U_{\sigma (t)}^T \otimes {I_d}} \right)x(t)$ and ${\tilde x_1}(t) = \left( {e_1^T \otimes {I_d}} \right)\tilde x(t)$, one can obtain that
\[{\tilde x_1}(t) = \left( {e_1^TU_{\sigma (t)}^T \otimes {I_d}} \right)x(t).\]
Since $e_1^TU_{\sigma (t)}^T = {{{{\bf{1}}^T}} \mathord{\left/{\vphantom {{{{\bf{1}}^T}} {\sqrt N }}} \right.\kern-\nulldelimiterspace} {\sqrt N }}$, one has
\[{\tilde x_1}(0) = \frac{1}{{\sqrt N }}\sum\limits_{i = 1}^N {{x_i}(0)} ,\]
which means that
\[\upsilon (0) = \frac{1}{N}\sum\limits_{i = 1}^N {{x_i}(0)} .\]
For simplicity of expression, one can choose that $c(t) = \upsilon (t).$ Thus, the conclusion of Theorem \ref{theorem1} can be obtained.\vspace{0em}
\end{proof}

\begin{rmk}\label{remark1}
For leaderless multi-agent networks, an interesting and challenging problem is to determine the consensus dynamics, which is usually described by the consensus function. Xiao and Wang first introduced the concept of the consensus function in \cite{c15}, where it was assumed that the consensus function is time-invariant. In \cite{c27}, an initial state projection method was proposed to determine the time-varying consensus function. In \cite{c15} and \cite{c27}, interaction topologies are fixed and the dynamics of each agent is linear, so the whole dynamics of multi-agent networks can be completely decomposed by the nonsingular transformation. However, their methods are no longer valid when each agent contains nonlinear dynamics and interaction topologies are switching. Theorem \ref{theorem1} shows that the consensus dynamics of nonlinear multi-agent networks also contains the nonlinear term and switching topologies do not impact the consensus function.\vspace{-2.0em}
\end{rmk}

\section{Guaranteed-cost consensualization criteria}\label{section4}\vspace{-0.5em}
In this section, the impacts of the nonlinear term $\left( {\bar U_{\sigma (t)}^T \otimes {I_d}} \right)F(x(t))$ in (\ref{10}) are linearized by using the structure property of the transformation matrix ${U_{\sigma (t)}}$ and the Lipschitz condition and guaranteed-cost consensualization criteria are presented to determine the gain matrix $K$. Furthermore, an approach to obtain the minimum guaranteed cost is proposed.\par

Let ${\lambda _{\min }} = \min \left\{ {{\lambda _{i,2}}{\rm{ }}\left( {i = 1,2, \cdots ,\eta } \right)} \right\}$ and ${\lambda _{\max }} = \max \left\{ {{\lambda _{i,N}}{\rm{ }}\left( {i = 1,2, \cdots ,\eta } \right)} \right\}$, then the following theorem establishes a sufficient condition for guaranteed-cost consensualization on the basis of the Riccati equation, where the dimension of the variable is independent of the number of agents.

\begin{thm} \label{theorem2}
Multi-agent network (\ref{1}) is guaranteed-cost consensualizable by protocol (\ref{2}) with $K = {B^T}P$ if there exists ${P^T} = P > 0$ such that
\[\Xi  = {A^T}P + PA + P\left( {\lambda _{\max }^2BR{B^T} - 2{\lambda _{\min }}B{B^T} + {I_d}} \right)P + 3{\lambda _{\max }}Q + {\gamma ^2}{I_d} = 0.\]
In this case, the guaranteed cost satisfies that
\[\beta  = \frac{1}{N}{x^T}(0)\left( {\left( {N{I_N} - {{\bf{1}}_N}{\bf{1}}_N^T} \right) \otimes P} \right)x(0).\]
\end{thm}
\begin{proof}
Consider the following Lyapunov function candidate
\[V(t) = {\varsigma ^T}(t)\left( {{I_{N - 1}} \otimes P} \right)\varsigma (t),\]
where $P$ is a solution of the Riccati equation $\Xi  = 0$. Taking the derivative of $V(t)$ with respect to $t$ along the solution of subsystem (\ref{10}), one can derive that
\begin{eqnarray}\label{13}
\begin{array}{l}
\dot V(t) = \sum\limits_{i = 2}^N {\tilde x_i^T(t)\left( {{A^T}P + PA - {\lambda _{\sigma (t),i}}{B^T}{K^T}P - {\lambda _{\sigma (t),i}}PBK} \right){{\tilde x}_i}(t)} \\ \hspace{33pt}
 + 2{\varsigma ^T}(t)\left( {\bar U_{\sigma (t)}^T \otimes P} \right)F(x(t)).
\end{array}
\end{eqnarray}
It can be shown that
\begin{eqnarray}\label{14}
2{\varsigma ^T}(t)\left( {\bar U_{\sigma (t)}^T \otimes P} \right)F(x(t)) \le \sum\limits_{i = 2}^N {\tilde x_i^T(t)PP{{\tilde x}_i}(t)}  + {F^T}(x(t))\left( {{{\bar U}_{\sigma (t)}}\bar U_{\sigma (t)}^T \otimes {I_d}} \right){\rm{ }}F(x(t)).
\end{eqnarray}
Due to ${U_{\sigma (t)}}U_{\sigma (t)}^T = {I_N}$ , one can show that
\[{\bar U_{\sigma (t)}}\bar U_{\sigma (t)}^T = \frac{1}{N}\left( {N{I_N} - {{\bf{1}}_N}{\bf{1}}_N^T} \right).\]
Thus, one has
\begin{eqnarray}\label{15}
{F^T}(x(t))\left( {{{\bar U}_{\sigma (t)}}\bar U_{\sigma (t)}^T \otimes {I_d}} \right){\rm{ }}F(x(t)) = \frac{1}{{2N}}\sum\limits_{i = 1}^N {\sum\limits_{j = 1}^N {{{\left\| {f({x_i}(t)) - f({x_j}(t))} \right\|}^2}} }.
\end{eqnarray}
Since
\[\left\| {f({x_i}(t)) - f({x_j}(t))} \right\| \le \gamma \left\| {{x_i}(t) - {x_j}(t)} \right\|{\rm{ }}\left( {i,j = 1,2, \cdots ,N} \right),\]
it can be obtained that
$$
\frac{1}{{2N}}\sum\limits_{i = 1}^N {\sum\limits_{j = 1}^N {{{\left\| {f({x_i}(t)) - f({x_j}(t))} \right\|}^2}} }  \le \frac{{{\gamma ^2}}}{{2N}}\sum\limits_{i = 1}^N {\sum\limits_{j = 1}^N {{{\left\| {{x_i}(t) - {x_j}(t)} \right\|}^2}} }
$$\vspace{-2em}
\begin{eqnarray}\label{16}
 \hspace{18em} = {\gamma ^2}{x^T}(t)\left( {{{\bar U}_{\sigma (t)}}\bar U_{\sigma (t)}^T \otimes {I_d}} \right){\rm{ }}x(t).
\end{eqnarray}
Due to
\[\bar U_{\sigma (t)}^T \otimes {I_d} = \left[ {0,{I_{(N - 1)d}}} \right]\left( {U_{\sigma (t)}^T \otimes {I_d}} \right),\]
it can be derived that
\begin{eqnarray}\label{17}
{\gamma ^2}{x^T}(t)\left( {{{\bar U}_{\sigma (t)}}\bar U_{\sigma (t)}^T \otimes {I_d}} \right){\rm{ }}x(t) = {\gamma ^2}\sum\limits_{i = 2}^N {\tilde x_i^T(t){{\tilde x}_i}(t)} .
\end{eqnarray}
From (\ref{13}) to (\ref{17}), one has
\begin{eqnarray}\label{18}
\dot V(t) \le \sum\limits_{i = 2}^N {\tilde x_i^T(t)\left( {{A^T}P + PA - {\lambda _{\sigma (t),i}}{K^T}{B^T}P - {\lambda _{\sigma (t),i}}PBK + PP + {\gamma ^2}{I_d}} \right){{\tilde x}_i}(t).}
\end{eqnarray}
Let $K = {B^T}P$, then it can be shown by $\Xi  = 0$ and (\ref{18}) that
\[\dot V(t) \le \sum\limits_{i = 2}^N {\tilde x_i^T(t)\left( {2({\lambda _{\min }} - {\lambda _{\sigma (t),i}})P{B^T}BP - \lambda _{\max }^2PBR{B^T}P - 3{\lambda _{\max }}Q} \right){{\tilde x}_i}(t)} .\]
Due to ${\lambda _{\min }} - {\lambda _{\sigma (t),i}} \le 0{\rm{ }}(i = 2,3, \cdots ,N)$ and ${\lambda _{\max }} > 0$, it can be derived that $\dot V(t) \le 0$ and $\dot V(t) \equiv 0$ if and only if ${\tilde x_i}(t) \equiv 0{\rm{ }}(i = 2,3, \cdots ,N)$; that is, ${\lim _{t \to {\rm{ + }}\infty }}\varsigma (t) = 0$. By the proof of Theorem \ref{theorem1}, multi-agent network (\ref{4}) achieves consensus.\par

Next, we analyze the guaranteed-cost performance of the gain matrix $K = {B^T}P$. It can be derived that
\begin{eqnarray}\label{19}
{J_{Cu}}(t) = {x^T}(t)\left( {L_{\sigma (t)}^2 \otimes PBR{B^T}P} \right)x(t),
\end{eqnarray}
\begin{eqnarray}\label{20}
{J_{Cx}}(t) = {x^T}(t)\left( {2{L_{\sigma (t)}} \otimes Q} \right)x(t).
\end{eqnarray}
Due to ${\lambda _{\sigma (t),1}} = 0$ and $\tilde x(t) = \left( {U_{\sigma (t)}^T \otimes {I_d}} \right)x(t)$, one has
\begin{eqnarray}\label{21}
{x^T}(t)\left( {L_{\sigma (t)}^2 \otimes PBR{B^T}P} \right)x(t) = \sum\limits_{i = 2}^N {\lambda _{\sigma (t),i}^2\tilde x_i^T(t)PBR{B^T}P{{\tilde x}_i}(t)} ,
\end{eqnarray}\vspace{-2em}
\begin{eqnarray}\label{22}
{x^T}(t)\left( {2{L_{\sigma (t)}} \otimes Q} \right)x(t) = \sum\limits_{i = 2}^N {2{\lambda _{\sigma (t),i}}\tilde x_i^T(t)Q{{\tilde x}_i}(t)} .
\end{eqnarray}
Let $T \ge 0$, then it can be obtained from (\ref{19}) to (\ref{22}) that
\[{J_T} \buildrel \Delta \over = \int_0^T {\left( {{J_{Cu}}(t) + {J_{Cx}}(t)} \right){\rm{d}}t}  = \sum\limits_{i = 2}^N {\int_0^T {\tilde x_i^T(t)\left( {2{\lambda _{\sigma (t),i}}Q + \lambda _{\sigma (t),i}^2PBR{B^T}P} \right){{\tilde x}_i}(t){\rm{d}}t} } .\]
Thus, one can derive that
\begin{eqnarray}\label{23}
\begin{array}{l}
{J_T} = \sum\limits_{i = 2}^N {\int_0^T {\tilde x_i^T(t)\left( {2{\lambda _{\sigma (t),i}}Q + \lambda _{\sigma (t),i}^2PBR{B^T}P} \right){{\tilde x}_i}(t){\rm{d}}t} }  + \int_0^T {\dot V(t)} {\rm{d}}t - V(T) + V(0)\\ \hspace{12pt}
{\rm{    }} \le \sum\limits_{i = 2}^N {\int_0^T {\tilde x_i^T(t)(2({\lambda _{\min }} - {\lambda _{\sigma (t),i}})P{B^T}BP + (\lambda _{\sigma (t),i}^2 - \lambda _{\max }^2)PBR{B^T}P} } \\ \hspace{23pt}
{\rm{       }} + (2{\lambda _{\sigma (t),i}} - 3{\lambda _{\max }})Q){{\tilde x}_i}(t){\rm{d}}t - V(T) + V(0)\\ \hspace{12pt}
{\rm{    }} \le V(0).
\end{array}
\end{eqnarray}
Due to ${\bar U_{\sigma (t)}}\bar U_{\sigma (t)}^T = {{\left( {N{I_N} - {{\bf{1}}_N}{\bf{1}}_N^T} \right)} \mathord{\left/
 {\vphantom {{\left( {N{I_N} - {{\bf{1}}_N}{\bf{1}}_N^T} \right)} N}} \right.
 \kern-\nulldelimiterspace} N}$ and $\varsigma (t) = \left[ {0,{I_{(N - 1)d}}} \right]\left( {U_{\sigma (t)}^T \otimes {I_d}} \right)x(t)$, it can be obtained that
\begin{eqnarray}\label{24}
{\varsigma ^T}(0)\left( {{I_{N - 1}} \otimes P} \right)\varsigma (0) = \frac{1}{N}{x^T}(0)\left( {\left( {N{I_N} - {{\bf{1}}_N}{\bf{1}}_N^T} \right) \otimes P} \right)x(0).
\end{eqnarray}
By (\ref{23}) and (\ref{24}), let $T \to  + \infty $, then the conclusion of Theorem \ref{theorem2} can be obtained.
\end{proof}

\begin{rmk}\label{remark2}
For linear multi-agent networks, by using an important property of the Laplacian matrix that its row sum is equal to zero, the dynamics of the whole network can be completely decomposed as the consensus and disagreement parts, which can used to determine the consensus function and the consensualization criteria, respectively. However, for nonlinear multi-agent networks, both the consensus and disagreement dynamics contain nonlinear terms as shown in (\ref{9}) and (\ref{10}). To linearize the influence of the nonlinear term, Theorem \ref{theorem2} applies the structure property of the time-varying part ${\bar U_{\sigma (t)}}$ of the orthonormal transformation matrix ${U_{\sigma (t)}}$; that is, ${\bar U_{\sigma (t)}}\bar U_{\sigma (t)}^T = {N^{ - 1}}\left( {N{I_N} - {{\bf{1}}_N}{\bf{1}}_N^T} \right)$. It should be pointed out that the influence of the nonlinear term cannot be dealt with by the methods in \cite{c22}-\cite{c24}, where some interesting and important guaranteed-cost consensualization results were proposed.
\end{rmk}

To ensure that $\Xi  = 0$ in Theorem \ref{theorem2} has a positive definite solution, ${A^T}P + PA - 2{\lambda _{\min }}B{B^T} < 0$ is necessary, which means that $(A,B)$ is stabilizable due to ${\lambda _{\min }} > 0$. Furthermore, these terms $PP + {\gamma ^2}{I_d}$ and $\lambda _{\max }^2PBR{B^T}P + 3{\lambda _{\max }}Q$ in $\Xi $ represent the impacts of the nonlinear dynamics and the cost function on guaranteed-cost consensualization, respectively. If ${\lambda _{R,\max }} < 2{\lambda _{\min }}\lambda _{\max }^{ - 2}$ with ${\lambda _{R,\max }} > 0$ denoting the maximum eigenvalue of the matrix $R$, then $\lambda _{\max }^2BR{B^T} - 2{\lambda _{\min }}B{B^T} < 0$. In this case, $\Xi  - PP - {\gamma ^2}{I_d} = 0$ has a unique and positive definite solution. Thus, when the nonlinear dynamics does not exist; that is, $f({x_i}(t)) \equiv 0$ $(i = 1,2, \cdots ,N)$, the following theorem presents a sufficient condition for guaranteed-cost consensualization, which is a more perfect result for guaranteed-cost consensus design than the one in \cite{c24}, where the existence and uniqueness of the solutions of LMI guaranteed-cost consensualization criteria for linear multi-agent networks with switching topologies cannot be ensured.\par

\begin{thm}\label{theorem3}
Multi-agent network (\ref{1}) with $f({x_i}(t)) \equiv 0$ $(i = 1,2, \cdots ,N)$ is guaranteed-cost consensualizable by protocol (\ref{2}) if ${\lambda _{R,\max }} < 2{\lambda _{\min }}\lambda _{\max }^{ - 2}$ and $(A,B)$ is stabilizable. In this case, the gain matrix is $K = {B^T}P$ with ${P^T} = P > 0$ the solution of ${A^T}P + PA + P\left( {\lambda _{\max }^2BR{B^T} - 2{\lambda _{\min }}B{B^T}} \right)P + 3{\lambda _{\max }}Q = 0$ and the guaranteed cost satisfies that $\beta  = {N^{ - 1}}{x^T}(0)\left( {\left( {N{I_N} - {{\bf{1}}_N}{\bf{1}}_N^T} \right) \otimes P} \right)x(0).$
\end{thm}

Moreover, it is a very interesting problem to determine the minimum guaranteed cost by choosing a proper matrix $P$. Because $\beta $ is associated with the initial states ${x_i}(0){\rm{ }}(i = 1,2, \cdots ,N)$, it is difficult to obtain the minimum guaranteed cost ${\beta ^ * }$. Here, the guaranteed cost is minimum in the sense that the initial state error between any two agents is a random variable with the zero mean value and $E\left\{ {({x_j}(0) - {x_i}(0)){{({x_j}(0) - {x_i}(0))}^T}} \right\} = {I_d}$, which is a similar assumption to optimization control of isolated systems as shown \cite{c28}. In this case, one can obtain that
\[
~~~~{\rm{  }}E\left\{ {\frac{1}{N}{x^T}(0)\left( {\left( {N{I_N} - {{\bf{1}}_N}{\bf{1}}_N^T} \right) \otimes P} \right)x(0)} \right\}
\]\vspace{-1.2em}
\[
\hspace{3.9em}= \frac{1}{{2N}}E\left\{ {\sum\limits_{i = 1}^N {\sum\limits_{j = 1}^N {{{\left( {{x_j}(0) - {x_i}(0)} \right)}^T}P\left( {{x_j}(0) - {x_i}(0)} \right)} } } \right\}
\]\vspace{-0.5em}
\[
 \hspace{-12em}= \frac{N}{2}{\rm{trace}}\left\{ P \right\}.
\]
Let ${\beta ^ * } = E\left\{ \beta  \right\}$, then ${\beta ^ * } = 0.5N \cdot {\rm{trace}}\left\{ P \right\}.$ From Theorem \ref{theorem2}, by Schur complement lemma in \cite{c29}, the following theorem presents an approach to determine the minimum guaranteed cost in terms of LMIs.

\begin{thm}\label{theorem4}
Multi-agent network (\ref{1}) is guaranteed-cost consensualizable by protocol (\ref{2}) with the minimum guaranteed cost if there exist ${\tilde P^T} = \tilde P > 0$ and ${\tilde X^T} = \tilde X > 0$ such that\par
\hspace{36pt} min ${\rm{trace}}\left( {\tilde X} \right)$\par
\hspace{36pt} s.t.
\[\begin{array}{l}
\left[ {\begin{array}{*{20}{c}}
{\tilde X}&{{I_d}}\\
*&{\tilde P}
\end{array}} \right] > 0,\\
\left[ {\begin{array}{*{20}{c}}
{\tilde P{A^T} + A\tilde P - \lambda _{\max }^2BR{B^T} - 2{\lambda _{\min }}B{B^T} + {I_d}}&{3{\lambda _{\max }}\tilde PQ}&{\gamma \tilde P}\\
*&{ - 3{\lambda _{\max }}Q}&0\\
*&*&{{-I_d}}
\end{array}} \right] < 0.
\end{array}\]
In this case, $K = {B^T}{\tilde P^{ - 1}}$ and the minimum guaranteed cost is ${\beta ^ * } = 0.5N \cdot {\rm{trace}}\left\{ {\tilde X} \right\}.$
\end{thm}

\begin{rmk}\label{remark3}
Intuitively speaking, the guaranteed cost $\beta $ increases as the number of agents increases since both $R$ and $Q$ are positive. However, the relationship between the guaranteed cost and the number of agents cannot be directly reflected by the explicit expression of the guaranteed cost $\beta $ in Theorem \ref{theorem2}, where both $N$ and ${N^{ - 1}}$ appear in $\beta $. Actually, the matrix ${N^{ - 1}}\left( {N{I_N} - {{\bf{1}}_N}{\bf{1}}_N^T} \right)$ is equivalent to the Laplacian matrix of a complete graph with the weights of all the edges ${N^{ - 1}}$, which has a simple zero eigenvalue and whose nonzero eigenvalues are one. Hence, the guaranteed cost is directly proportional to the number of agents as $x(0)$ and $P$ are given, which is also coincident with the conclusion of Theorem \ref{theorem4} when the minimum guaranteed cost is considered.
\end{rmk}

\section{Numerical simulations}\label{section5}
Consider a Lipschitz nonlinear multi-agent network with six agents and the dynamics of each agent is described by (\ref{1}) with
\[A = \left[ {\begin{array}{*{20}{c}}
0&1&0&0\\
{ - 48.60}&{ - 1.25}&{48.60}&0\\
0&0&0&1\\
{19.50}&0&{ - 19.50}&0
\end{array}} \right],
B = \left[ \begin{array}{c}
0\\
21.60\\
0\\
0
\end{array} \right],
f({x_i}) = \left[ \begin{array}{c}
0\\
0\\
0\\
 - \gamma \sin ({x_{i3}})
\end{array} \right],\]
where ${x_i} = {\left[ {{x_{i1}},{\rm{ }}{x_{i2}},{\rm{ }}{x_{i3}},{\rm{ }}{x_{i4}}} \right]^T}$ with $i = 1,2,...,6$ and $\gamma  = 0.333$. Figure 1 gives four undirected interaction topologies in the switching set and the switching signal $\sigma (t)$ is shown in Figure 2.\par

\begin{figure}[!htb]
\begin{center}
\scalebox{0.9}[0.9]{\includegraphics{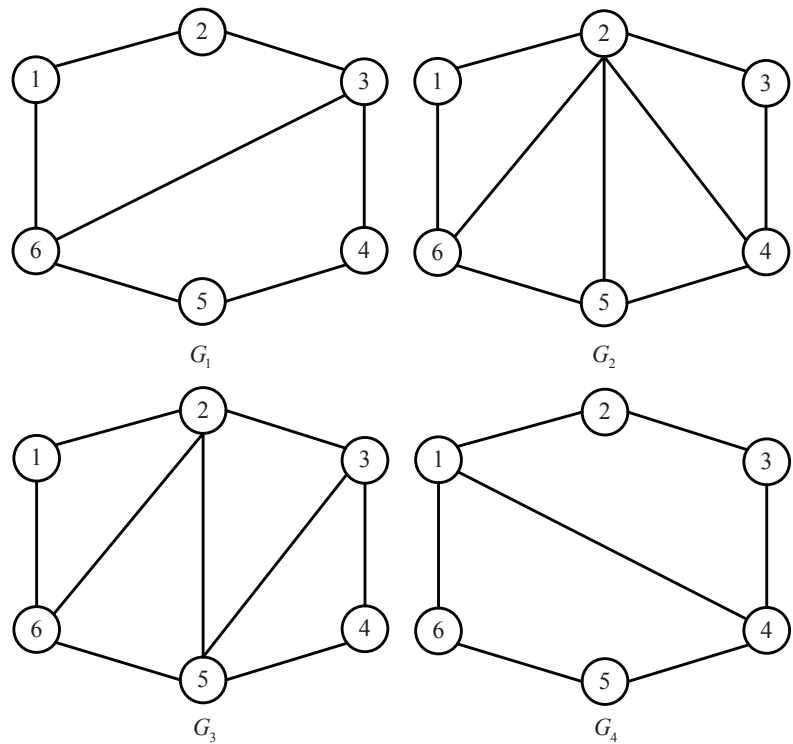}}
\vspace{0em}
\caption{Switching topology set.}
\end{center}\vspace{0em}
\end{figure}

\begin{figure}[!htb]
\begin{center}
\scalebox{0.5}[0.5]{\includegraphics{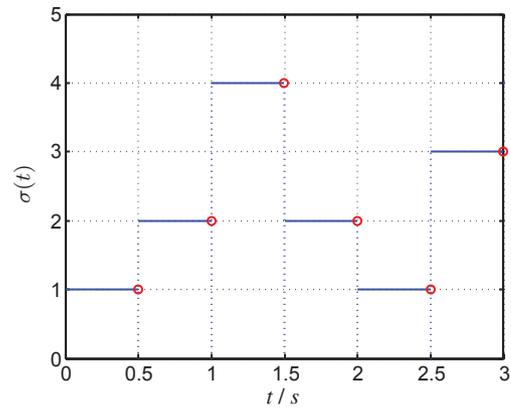}} \put (-200, 70)
{\rotatebox{90} {{\scriptsize $\sigma (t)$}}} \put (-108, -5) {{ \scriptsize {\it t}~/~\it s}}
\vspace{0em} \caption{Switching signal $\sigma (t)$.}
\end{center}\vspace{-2em}
\end{figure}

\begin{figure}[!htb]
\begin{center}
\scalebox{0.45}[0.45]{\includegraphics{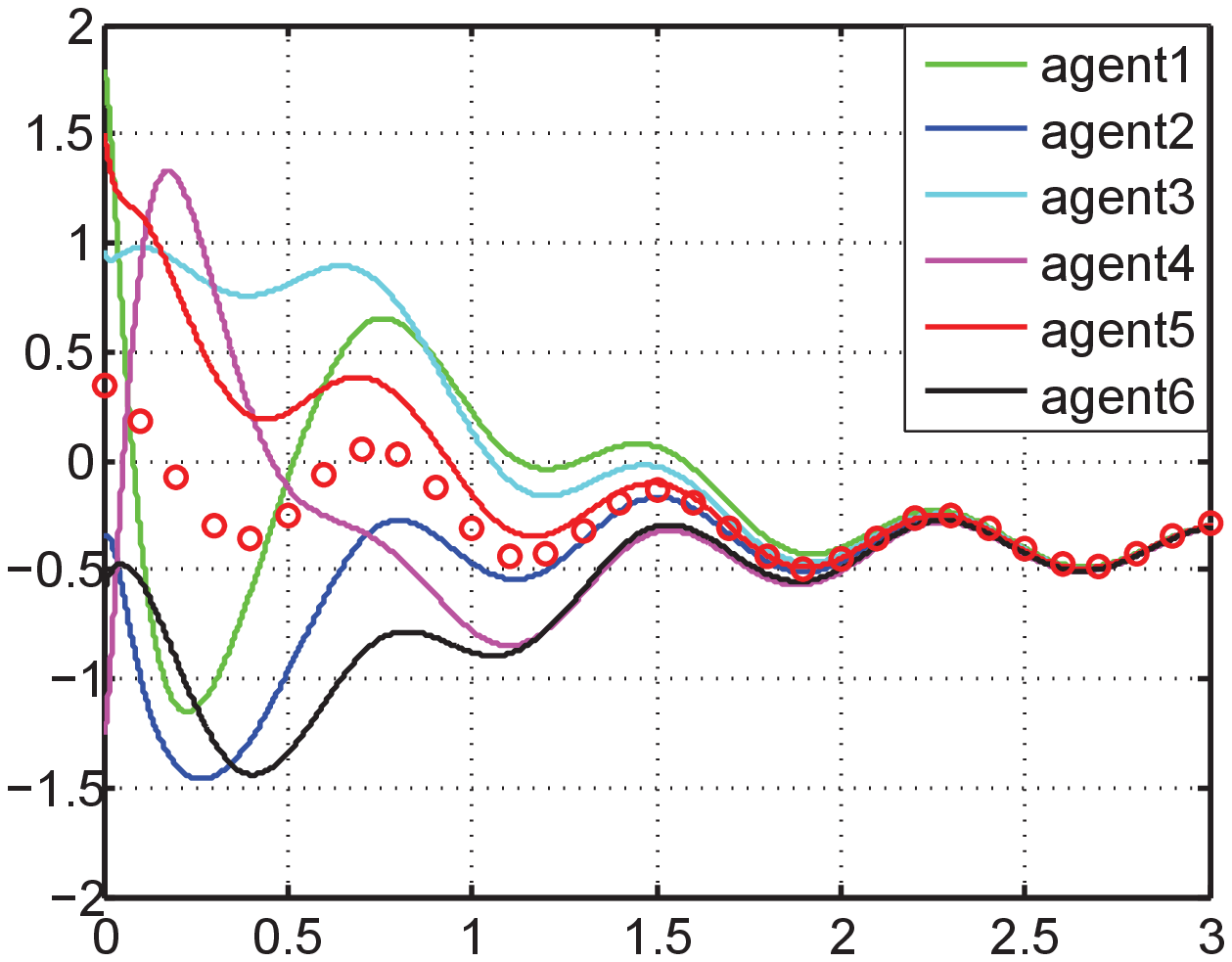}}
\scalebox{0.45}[0.45]{\includegraphics{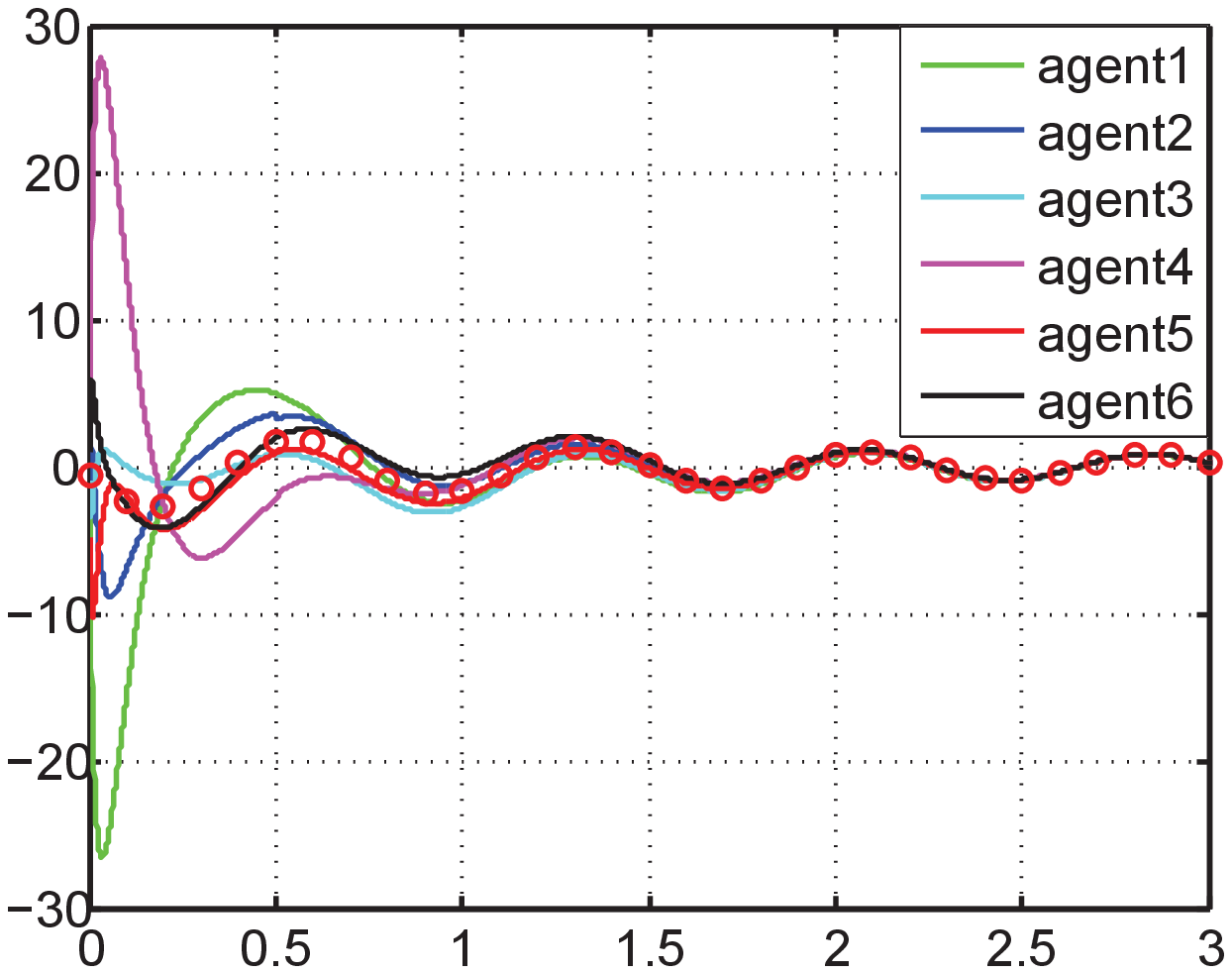}}
\put (-370, 60)
{\rotatebox{90} {{\scriptsize ${x_1}(t)$}}} \put (-280, -5) {{ \scriptsize {\it t}~/~\it s}}
\put (-185, 60)
{\rotatebox{90} {{\scriptsize ${x_2}(t)$}}} \put (-97, -5) {{ \scriptsize {\it t}~/~\it s}}

\end{center}\vspace{-2em}
\end{figure}

\begin{figure}[!htb]
\begin{center}
\scalebox{0.45}[0.45]{\includegraphics{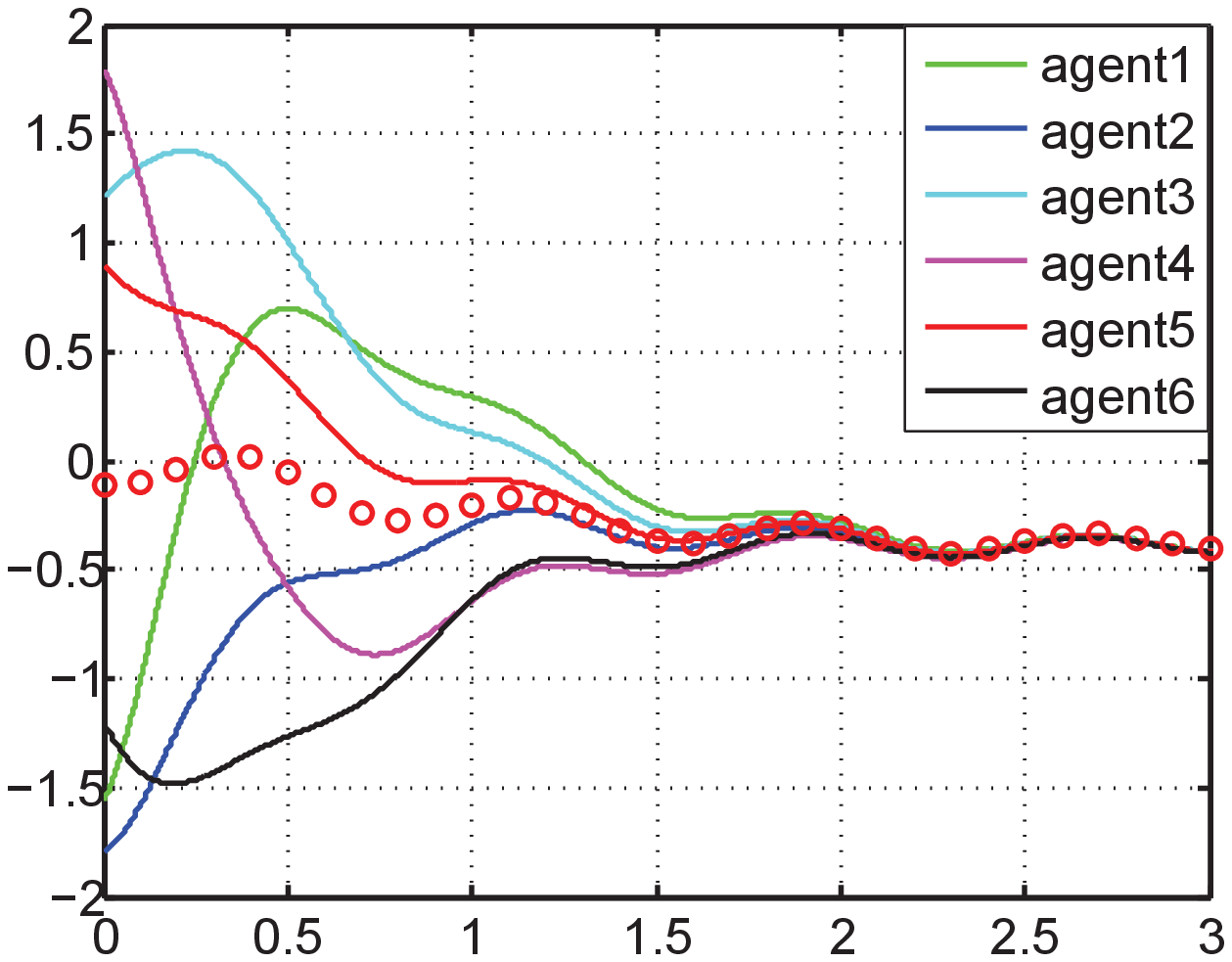}}
\scalebox{0.45}[0.45]{\includegraphics{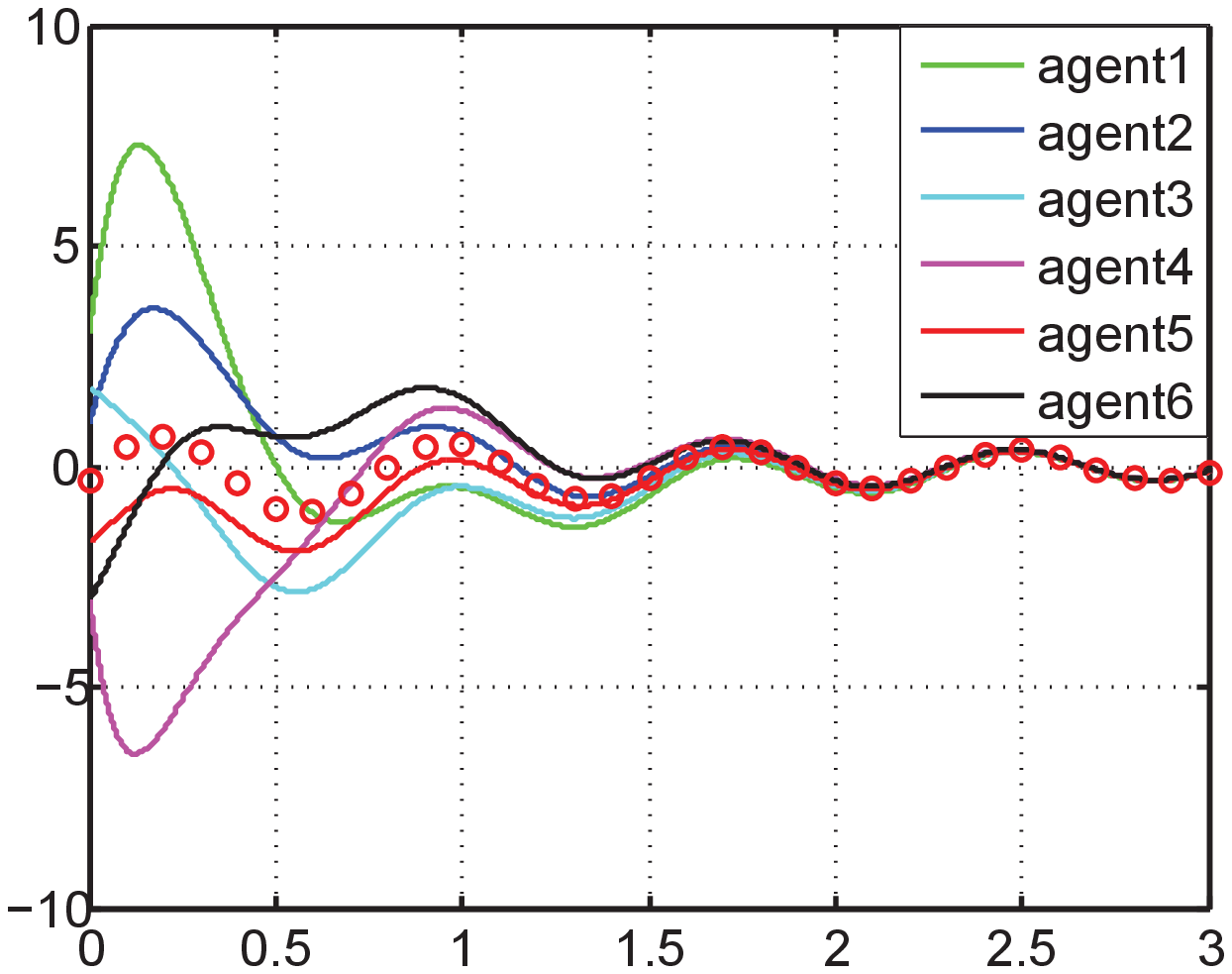}}
\put (-370, 60)
{\rotatebox{90} {{\scriptsize ${x_3}(t)$}}} \put (-280, -5) {{ \scriptsize {\it t}~/~\it s}}
\put (-185, 60)
{\rotatebox{90} {{\scriptsize ${x_4}(t)$}}} \put (-97, -5) {{ \scriptsize {\it t}~/~\it s}}
\vspace{0em} \caption{Output trajectories.}
\end{center}\vspace{-2em}
\end{figure}

\begin{figure}[!htb]
\begin{center}
\scalebox{0.5}[0.5]{\includegraphics{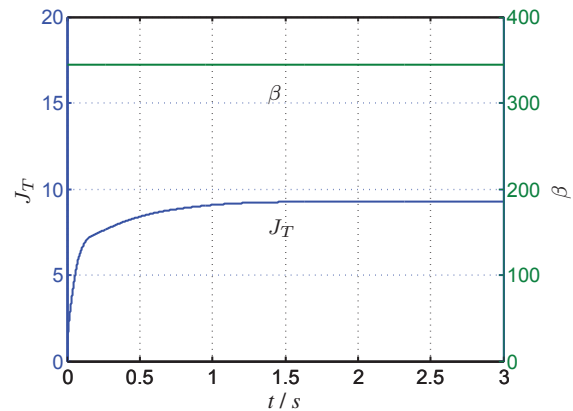}} \put (-2, 73)
{\rotatebox{90} {{\scriptsize $\beta $}}} \put (-203, 71)
{\rotatebox{90} {{\scriptsize $J_T $}}} \put (-111, -5) {{
\scriptsize {\it t}~/~\it s}} \put (-111, 60) {{ \scriptsize $J_T
$}} \put (-111, 110) {{ \scriptsize $\beta $}}
 \vspace{0em}
\caption{Trajectory of the cost function $J_T $.}
\end{center}\vspace{-3em}
\end{figure}

In Figure 3, the state trajectories of this multi-agent network are shown, and the trajectories marked by circles are the curves of the consensus function given in Theorem \ref{theorem1} and the full curves represent trajectories of states of six agents. Figure 4 depicts the trajectory of the cost function ${J_T}$. One can see from Figures 3 and 4 that state trajectories of all agents converge to the ones marked by circles and the cost function ${J_T}$ converges to a finite value less than $\beta $. The simulation results illustrate that this Lipschitz nonlinear multi-agent network achieves guaranteed-cost consensus. \vspace{-1.5em}

\section{Conclusions}\label{section6}\vspace{-0.5em}
The current paper addressed guaranteed-cost consensus for high-order multi-agent networks with Lipschitz nonlinearity and switching topologies. An explicit expression of the consensus dynamics, which is intrinsically coincident with the own dynamics of each agent, was given and it was shown that the initial state of the consensus dynamics is the average of initial states of all agents. Furthermore, based on the Riccati equation, a guaranteed-cost consensualization criterion was presented, where the dimension of the variable is independent of the number of agents, and an upper bound of the guaranteed cost was determined, which is associated with the initial states of all agents. Moreover, an approach to determine the minimum guaranteed cost was proposed in terms of LMIs.

The current paper assumes that all interaction topologies in the switching set are undirected and connected, which means that the Laplacian matrix associated with each topology is symmetric and has a simple zero eigenvalue. By this structure feature, the energy consumption term and the consensus regulation term can be simplified as shown in (\ref{21}) and (\ref{22}), and the impacts of the nonlinear term in the disagreement dynamics can be eliminated as shown in (\ref{16}) and (\ref{17}). However, the Laplacian matrix of a directed topology may be asymmetric and may not be diagonalizable, so the approaches in the current paper are no longer valid and some new methods are required to simplify the quadratic cost function and eliminate the impacts of the nonlinear term. We will focus on guaranteed-cost consensus control for nonlinear multi-agent networks with directed switching topologies in our further work.\vspace{-1.5em}

\section{Appendix}\label{section7}\vspace{-0.5em}
An undirected graph $G$ consists of a node set $V(G) = \left\{ {{v_1},{v_2}, \cdots ,{v_N}} \right\}$, an edge set $E(G) \subseteq \left\{ {\left( {{v_i},{v_j}} \right):{v_i},{v_j} \in V(G)} \right\}$, and a symmetric adjacency matrix $W = \left[ {{w_{ij}}} \right] \in {\mathbb{R}^{N \times N}}$, where ${w_{ij}} > 0$ if $\left( {{v_j},{v_i}} \right) \in E(G)$ and ${w_{ii}} = 0$~$\left( {i = 1,2, \cdots ,N} \right)$. The in-degree matrix and the Laplacian matrix of $G$ is defined as $D = {\rm{diag}}\left\{ {\sum\nolimits_{j \in {N_i}} {{w_{ij}}} ,i = 1,2, \cdots ,N} \right\}$ and $L = D - W$, respectively, where ${N_i} = \left\{ {{v_j} \in } \right.\left. {V(G):\left( {{v_j},{v_i}} \right) \in E(G)} \right\}$ denotes the neighbor set of node ${v_i}$. $G$ is said to be connected if there exists at least an undirected path between any two nodes. It can be shown that $L$ at least has a zero eigenvalue and ${{\bf{1}}_N}$ is the associated eigenvector. Especially, $0$ is a simple eigenvalue of $L$ and all the other $N-1$ eigenvalues are positive if $G$ is connected. More details on graph theory can be found in \cite{c30}.
\end{document}